\documentclass[aps,prl,twocolumn,superscriptaddress,floatfix,longbibliography]{revtex4-2}

\usepackage{amsmath,amssymb,bm,mathtools}
\usepackage{graphicx}
\usepackage{dsfont}
\usepackage[colorlinks=true,allcolors=blue]{hyperref}

\newcommand{\E}{\mathbb{E}}

\newcommand{\Dchi}{D_{\chi^2}}
\newcommand{\DRF}{D_{\mathrm{RF}}}
\newcommand{\ktrue}{g^{\mathrm{true}}}

\usepackage{amsthm}

\newtheorem{lemma}{Lemma}
\newtheorem{theorem}{Theorem}
\newtheorem{proposition}{Proposition}

\begin{document}

\title{State convertibility and fluctuation theorems from a dynamical reference: majorization meets martingales}

\author{Davide Cugini}
\affiliation{Dipartimento di Fisica ``A. Volta,'' Universit\`a di Pavia, via Bassi 6, 27100 Pavia, Italy}
\author{Giacomo Guarnieri}
\affiliation{Dipartimento di Fisica ``A. Volta,'' Universit\`a di Pavia, via Bassi 6, 27100 Pavia, Italy}
\affiliation{INFN Sezione di Pavia, Via Agostino Bassi 6, I-27100 Pavia, Italy}

\date{\today}

\begin{abstract}
State convertibility represents a fundamental concept used to determine whether a transformation is possible given a specific set of resources. Within the field of Thermodynamics, where physical process are required to preserve a reference state typically in microcanonical or canonical form, this translates into the notions of majorization and thermo-majorization ---criteria that require constructing and comparing state-dependent Lorenz
curves. 
In this work, we firstly unify and extend these notions to an arbitrary and possibly
time-dependent reference distribution $g(t)$, introducing the concept of $g(t)$-majorization; we then introduce a dual picture whereby state convertibility is turned into a one-dimensional
convex-order problem, which allows us to demonstrate that a transition is admissible if and only if the associated real-valued distributions of relative populations $ k_j(t)/g_j(t)$ are connected by a martingale. 
Building on it, we then derive an exact
fluctuation theorem for a reference-relative entropy production whose average
violation certifies, through a $\chi^{2}$-divergence bound, the mismatch
between an assumed and the true reference evolution---a model-independent
diagnostic that requires no independent characterization of the latter and
turns an observed breakdown of the fluctuation relation into a certified lower
bound on the reference error.
\end{abstract}

\maketitle

\emph{Introduction.---}%
The cornerstone of the entire statistical mechanics, the probabilistic and microscopic framework of Thermodynamics, was laid by Ludwig Boltzmann in 1877~\cite{Boltzmann1877}. 
In this pioneering paper, he showed that for $N$ particles at fixed total energy, a given energy-distribution profile $k=(k_1,\ldots, k_m)$, with $k_j=n_j/N$, is compatible with $W(k)=N!/\prod_j (N k_j)!$ microscopic configurations. If the particles are then drawn from a uniform probability distribution, Boltzmann showed that the occurrence of any energy-distribution is observed with probability proportional to $W(k)$
and introduced the notion of statistical entropy as the logarithm of this count. By then identifying equilibrium as simply the most probable configuration, Boltzmann turned Thermodynamics into a counting problem and gave the Second Law its statistical meaning.
In terms of
the occupation fractions $k_j=n_j/N$, Stirling's formula renders the counting
entropy as $\tfrac1N\ln W\to-\sum_j k_j\ln k_j$, the form later central to Gibbs
and, in information theory, to Shannon~\cite{Shannon1948}. The object at the
center of the argument is the \emph{population vector} $k=(k_1,\dots,k_m)$, a
normalized distribution over microstates.\\\indent
Despite its fundamental importance, Boltzmann's original construction was still subject to limitation: it
singled out the unique most probable macrostate relative to a fixed, uniform
count, and it quantified only the average of the underlying fluctuating quantity, not
its distribution, nor the transformations between macrostates that a physical
process may actually perform. The century that followed resolved these
limitations, but in separate directions, each generalizing the population
vector in its own way and each leaving a characteristic blind spot. Large
deviation theory recognized Boltzmann's counting as the first instance of a
large-deviation principle~\cite{Ellis1999,Touchette2009} and, through the work
of Kullback and Leibler in statistics~\cite{KullbackLeibler1951} and Sanov in
probability~\cite{Sanov1957}, generalized the entropy to the relative entropy
$D_{\mathrm{KL}}(k\|g)=\sum_j k_j\ln(k_j/g_j)$ against an arbitrary reference
$g$, identifying it as the exponential rate at which the empirical distribution
of samples drawn from $g$ deviates to $k$.
Information theory, in parallel, recognized that this same
relative entropy coincides with the negative of Boltzmann's entropy and with
Shannon's information entropy~\cite{Shannon1948,CoverThomas}, and, following
Jaynes~\cite{Jaynes1957}, recast equilibrium itself as inference against a
reference---supplying the bridge between thermodynamics and information, but not
the dynamical question of which state transformations are physically allowed.
That question is the leitmotif of the resource-theoretical approach to Thermodynamics,
which established that convertibility is governed by majorization when the
allowed operations preserve a uniform distribution, and by thermo-majorization
when they preserve a Gibbs
state~\cite{HardyLittlewoodPolya,Nielsen1999,Horodecki2013,Lostaglio2016}; yet
these criteria are tied to a \emph{fixed}, equilibrium reference and are decided
by constructing and comparing state-dependent Lorenz curves, an operationally
demanding procedure. A fourth thread, martingale theory, meanwhile turned the
Second Law into a family of exact fluctuation relations for the entropy produced
along individual trajectories~\cite{Jarzynski1997,Crooks1999,Seifert2012} and
revealed that the exponential of the entropy production is a
martingale~\cite{Neri2017,Roldan2023,Manzano2019}---but, again, against the
equilibrium or steady-state reference of standard stochastic thermodynamics.

Each of these developments takes a reference distribution as the object against
which irreversibility is measured, exactly as Boltzmann took his count; and each
treats a different facet---the rate of a fluctuation, its informational meaning,
the convertibility it permits, its trajectory statistics---in a formalism of its
own. What has been missing is the common structure that would let a single,
arbitrary, and possibly time-dependent reference $g(t)$ carry all four at once.\\\indent
Our work precisely provides that. Associating to each population vector $k(t)$ the real-valued
distribution of its relative populations $\kappa_j=k_j(t)/g_j(t)$---the
\emph{parent distribution}---collapses the separate questions into one. State
convertibility, classically decided by comparing Lorenz curves, becomes a
one-dimensional convex-order comparison of parent distributions, equivalent to
the existence of a martingale connecting them. 
The pointwise quantity $-\ln\kappa_j$---the
reference-relative surprisal of information theory---defines a
\emph{reference-relative entropy production} $\Delta s$ whose mean is
Boltzmann's relative entropy and which obeys an exact integral fluctuation
theorem, placing the fluctuation relations of stochastic thermodynamics within
martingale theory for an arbitrary dynamical reference. 
In addition, we also show that when the assumed
reference departs from the true one, the average violation of that theorem
certifies, through a $\chi^2$ divergence, a model-independent lower bound on the
mismatch; we illustrate this last result on a driven two-level system, where an
apparent breakdown of the second law becomes a quantitative measure of
nonadiabatic driving. Large deviation theory, information theory, resource
theories, and stochastic fluctuation relations thus meet on the single terrain
that Boltzmann first surveyed: the distribution of a system's populations
relative to a reference.\\

\emph{State convertibility.---}%
Consider a generic set of microstates $\{\mathcal{S}_j\}_{j=1}^m$ (not necessarily associated with particles' energies)
and let $g(t) = (g_1(t),\,...g_m(t))$ be a reference distribution of their populations whose time-evolution is known a priori.
Moreover, 
let $k(t)$ be another distribution absolutely continuous with
respect to $g(t)$, i.e. $g_j=0\Rightarrow k_j=0$ (automatic for $g_j>0$). We ask whether,
among the stochastic processes compatible with the evolution of $g$, there
exists one carrying $k(t_1)$ into a \textit{target} $k(t_2)$: that is, whether there is a
column-stochastic \emph{transition matrix} $T$ such that
\begin{equation}
T\,g(t_1)=g(t_2)\quad\wedge\quad T\,k(t_1)=k(t_2).
\label{eq:convertibility}
\end{equation}
The first condition fixes $T$ to be a legitimate evolution for the reference
itself; the second demands that the same $T$ reproduce the desired transition.
Convertibility is thus a statement about compatibility with the reference
dynamics, not about $k(t_1)\!\to\! k(t_2)$ in isolation.

The central object is the vector of \emph{relative populations}
\begin{equation}
\kappa_j(t):=\frac{k_j(t)}{g_j(t)},
\label{eq:kappa}
\end{equation}
which measures, microstate by microstate, how far $k$ deviates from the
reference: $\kappa_j=1$ signals a microstate populated exactly as $g$ predicts,
$\kappa_j\gtrless 1$ over- or under-population.
A first characterization of admissibility can be formulated directly in terms of the parent distribution, through a Lorenz-type integral inequality that unifies and generalizes majorization and thermo-majorization to an arbitrary dynamical reference, which we therefore term \emph{$g(t)$-majorization}.
While important for the complete picture, it is not needed for the main results that follow and therefore deferred to the End Matter [see Eq.~\eqref{eq:lorenz}]. 

Here we proceed directly to the transparent probabilistic criterion.
To each $\kappa(t)$ we associate
the real-valued \emph{parent distribution}
\begin{equation}
p_\kappa(x;t)=\sum_{j}g_j(t)\,\delta\!\big(x-\kappa_j(t)\big),
\label{eq:parent}
\end{equation}
the law of the relative population obtained by drawing a microstate from the
reference ensemble $g(t)$: it tracks the probability, under $g$, of drawing a
given amount of deviation from the reference. Because $g$ and $k$ are each
normalized, $p_\kappa$ has unit mass and unit first moment,
$\E_{x\sim p_\kappa}[1]=\sum_j g_j=1$ and
$\E_{x\sim p_\kappa}[x]=\sum_j g_j\kappa_j=\sum_j k_j=1$; it is thus a genuine
probability distribution, but one over relative populations rather than over
microscopic configurations.

The equivalence rests on a coupling. A martingale coupling of the two parent
distributions is a joint law $P(n,j):= \mathrm{P(x_1 = \kappa_n(t_1),x_2= \kappa_j(t_1))}$, with
$x_1\sim p_\kappa(\cdot;t_1)$ and $x_2\sim p_\kappa(\cdot;t_2)$, whose marginals
are the reference weights,
$\sum_j P(n,j)=g_n(t_1)$ and $\sum_n P(n,j)=g_j(t_2)$, and which satisfies the
martingale condition $\E[x_1\mid x_2=\kappa_j(t_2)]=\kappa_j(t_2)$. Writing the
conditional expectation explicitly and using Bayes' rule, the martingale
condition reads $\sum_n \kappa_n(t_1)\,P(n,j)=\kappa_j(t_2)\,g_j(t_2)=k_j(t_2)$.
Defining
\begin{equation}
T_{jn}:=\frac{P(n,j)}{g_n(t_1)},
\label{eq:Tdef}
\end{equation}
the marginals and the martingale condition translate, respectively, into
column-stochasticity of $T$, reference preservation $T g(t_1)=g(t_2)$, and the
physical transition $T k(t_1)=k(t_2)$; the map~\eqref{eq:Tdef} is a bijection,
so the converse holds as well~\cite{SM}. This establishes our first main result.

\begin{theorem}[Convertibility criterion]\label{res:convert}
A transition $k(t_1)\!\to\!k(t_2)$ compatible with the reference dynamics of
$g(t)$ exists if and only if the associated parent distributions are connected by a martingale.
Namely
\begin{align}
\exists\,T|\ Tg(t_1)=g(t_2) \wedge Tk(t_1)=k(t_2)
\end{align}
iff
\begin{align}
p_\kappa(\cdot\,;t_1)\xrightarrow{\,\mathrm{mart.}\,} p_\kappa(\cdot\,;t_2).\nonumber
\end{align}
\end{theorem}
\vspace{0.3cm}
\noindent
Physically, the transition is possible precisely when the distribution of
relative populations \emph{concentrates} in time---when the population profile
relaxes toward the reference $g$, so that $p_\kappa$ narrows; transitions that
would require it to spread are forbidden. 
By Strassen's theorem~\cite{Strassen1965}, the existence of a martingale
coupling is equivalent to the convex order of the two parent distributions,
$p_\kappa(\cdot;t_1)\overset{\mathrm{cx}}{\succeq}p_\kappa(\cdot;t_2)$, that is,
\begin{equation}
\E_{x\sim p_\kappa(\cdot;t_1)}[\phi(x)]
\ \ge\
\E_{x\sim p_\kappa(\cdot;t_2)}[\phi(x)]
\label{eq:convexorder}
\end{equation}
for every convex $\phi$ for which the expectations are defined. Since all
parent distributions share unit mass and unit mean, equality holds
automatically for affine $\phi$, as the definition of convex order requires;
the content of \eqref{eq:convexorder} is therefore the ordering on the
genuinely convex part. Convex order thus plays, in the space of parent
distributions, the role majorization plays in the space of population
vectors---an equivalence classically known for a static, uniform
reference~\cite{Cartier1964,MarshallOlkin} and extended here, through the
parent-distribution construction, to an arbitrary dynamical reference $g(t)$.
The $m$-dimensional, sorting-based comparison of population vectors has
become a one-dimensional convex-order comparison of real-valued
distributions, decidable by the existence of a martingale. This
distributional viewpoint is closely related to the correspondence between
ranked samples and their associated order-statistics distributions developed
in ~\cite{Cugini2025}, where ranking information is likewise recast in
terms of the statistics of an underlying real-valued distribution.
That convertibility should be judged relative to a reference distribution is
the organizing principle of the resource-theoretic view of
thermodynamics~\cite{Brandao2015,Chitambar2019} and, classically, of
$d$-majorization and relative majorization and their roots in the comparison of
statistical experiments~\cite{Blackwell1953,Veinott1971,Ruch1978,Ruch1980};
the present construction realizes that principle for a reference that is itself
dynamical.
The chain of equivalences
$\eqref{eq:convertibility}\Leftrightarrow\eqref{eq:lorenz}
\Leftrightarrow\text{martingale}\Leftrightarrow\eqref{eq:convexorder}$
is summarized in Fig.~\ref{fig:schematic}; the sorting-based construction of
$\Gamma_\kappa$ has been replaced by a single one-dimensional convex-order
comparison, checkable through a martingale coupling.\\

\begin{figure}[t]
\centering
\includegraphics[width=\linewidth]{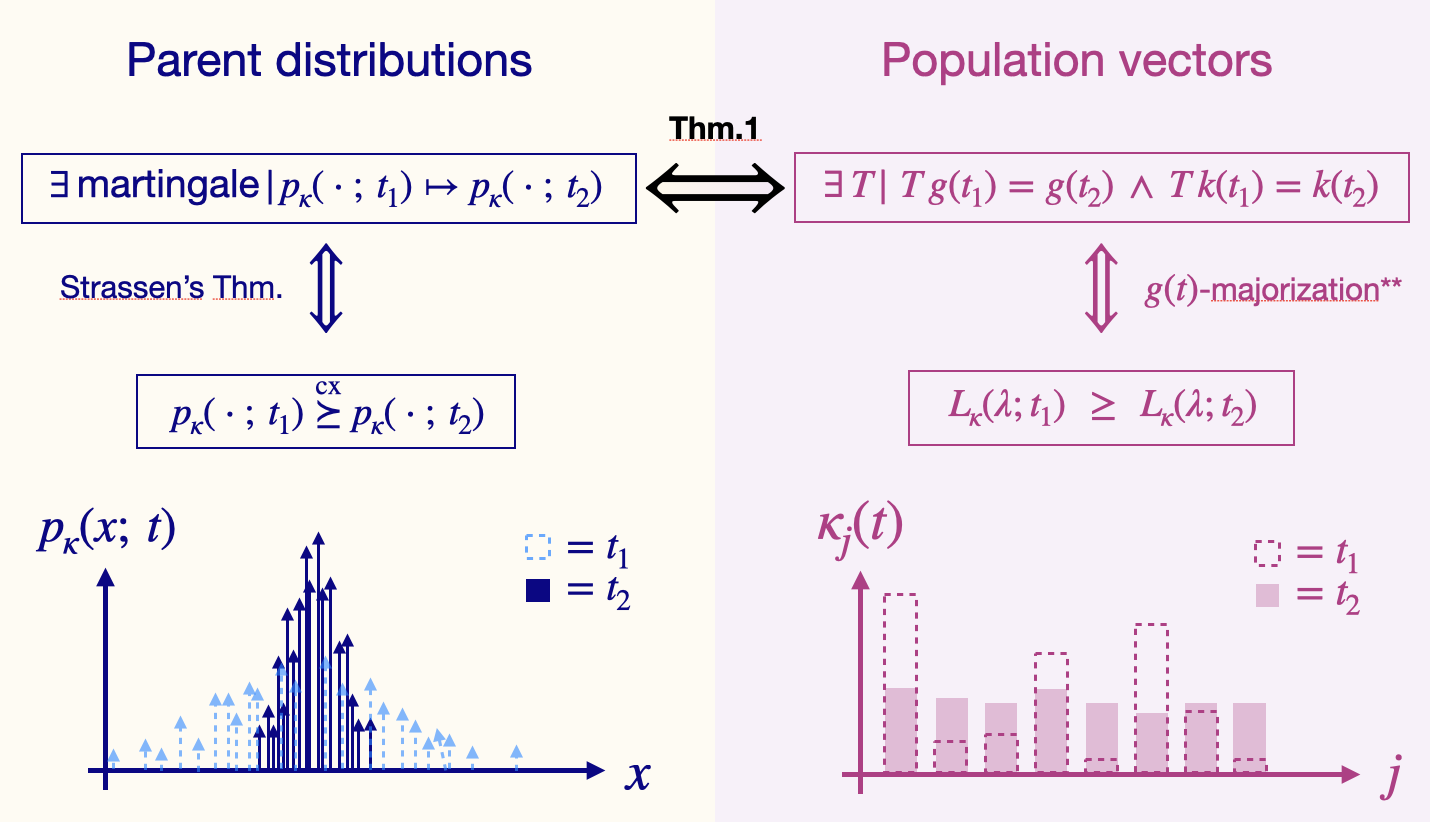}
\caption{Schematic of the framework equivalence demonstrated in this paper. 
A population vector $k(t)$, compared
against a reference $g(t)$, is mapped to the parent distribution $p_\kappa(\cdot;\,t)$ of
its relative populations $\kappa_j(t)=k_j(t)/g_j(t)$. State convertibility
\eqref{eq:convertibility} and the $g(t)$-majorization criterion ($^{**}$see End Matter \eqref{eq:lorenz}) on the right side (purple background) are shown to be fully equivalent to each other and to the existence of a martingale coupling between the respective parent distributions (Theorem~\ref{res:convert}) and finally, by virtue of Strassen's Theorem, to convex order relation \eqref{eq:convexorder} (left hand side). 
}
\label{fig:schematic}
\end{figure}

\emph{Fluctuation relation.---}%
Beyond convertibility, the same reference-state construction yields a natural
fluctuation relation. For each microstate $\mathcal{S}_j$ define the reference-relative
surprisal
\begin{equation}
s_j(t):=-\ln\kappa_j(t)=-\ln\frac{k_j(t)}{g_j(t)},
\label{eq:geoentropy}
\end{equation}
so that $e^{-s_j(t)}=\kappa_j(t)$. This is the difference between the Shannon
information content $-\ln k_j$ of microstate $j$ and that assigned to it by the
reference, $-\ln g_j$~\cite{Shannon1948,CoverThomas}; sampled according to the
population vector, its mean is precisely the relative entropy,
$\sum_j k_j s_j=D_{\mathrm{KL}}(k\|g)$, recovering Boltzmann's and Sanov's
entropy. 
For a process $T$ which now is assumed to evolve the population vector $k(t_1)$ 
into the population vector $k(t_2) = T\,k(t_1)$,
we introduce the \emph{reference-relative entropy production} as
$\Delta s:=s_{j_2}(t_2)-s_{j_1}(t_1)$, and its average is taken over the
physical joint law of $(j_1,j_2)$, 
\begin{equation}
\Pi(j_1,j_2)=k_{j_1}(t_1)\,T_{j_2 j_1},
\label{eq:Pi}
\end{equation}
the initial population followed by the forward transition. This is the law of where the
system actually is at the two times, and it must be distinguished from the
reference-weighted coupling of the previous section. This yields our second
main result.

\begin{theorem}[Reference-generalized fluctuation theorem]\label{res:ift}
For any transition matrix $T$ compatible with the reference dynamics,
$Tg(t_1)=g(t_2)$, the reference-relative entropy production satisfies the
integral fluctuation theorem
\begin{equation}
\E_\Pi\!\big[e^{-\Delta s}\big]=1,
\qquad\text{whence}\qquad
\E_\Pi[\Delta s]\ge 0 .
\label{eq:ift}
\end{equation}
\end{theorem}

\noindent
The proof is immediate: using $k_{j_1}(t_1)=g_{j_1}(t_1)\kappa_{j_1}(t_1)$ and
then reference preservation $Tg(t_1)=g(t_2)$,
\begin{align}
\E_\Pi\!\big[e^{-\Delta s}\big]
&=\sum_{j_1,j_2}k_{j_1}(t_1)T_{j_2 j_1}\frac{\kappa_{j_2}(t_2)}{\kappa_{j_1}(t_1)}
\\\notag&=\sum_{j_2}\kappa_{j_2}(t_2)\,g_{j_2}(t_2)
=\sum_{j_2}k_{j_2}(t_2)=1,
\label{eq:iftproof}
\end{align}
and $\E_\Pi[\Delta s]\ge0$ follows from Jensen's inequality applied to the
convex function $e^{-x}$.

Equation~\eqref{eq:ift} is the martingale-driven generalization, to an
arbitrary and possibly time-dependent reference $g(t)$, of the integral
fluctuation theorems of stochastic thermodynamics, ordinarily stated relative
to a Gibbs reference~\cite{Jarzynski1997,Crooks1999}, and it is consistent with
the exponential-martingale structure of entropy
production~\cite{Neri2017,Roldan2023,Manzano2019,Manzano2021}. The second-law
inequality $\E_\Pi[\Delta s]\ge0$ it entails is a statement about
irreversibility measured purely with respect to the chosen reference: no
admissible process decreases the reference-relative entropy on average.\\

\emph{Certifying reference mismatch.---}%
Equation~\eqref{eq:ift} presumes the evolution of $g(t)$ is known exactly. In
practice it often is not. Consider a system driven quasi-adiabatically by a
time-dependent Hamiltonian $H(t)$ in contact with a bath at inverse temperature
$\beta$; initialized in the Gibbs state $g(t_1)\propto e^{-\beta H(t_1)}$, its
state at $t_2$ is only approximated by the instantaneous Gibbs state
$g(t_2)\propto e^{-\beta H(t_2)}$, the approximation becoming exact only in the
adiabatic limit. More generally the assumed reference $g(t_2)$ may differ from
the true evolution $\ktrue(t_2)$ of $g(t_1)$, and one wishes to quantify the
discrepancy.

Suppose, then, that the physical process is compatible with the true reference,
$T g(t_1)=\ktrue(t_2)$, while the entropy production is evaluated with the
assumed reference $g(t_2)$. Repeating the derivation of \eqref{eq:ift} with this
single change, the inner sum collapses to $\ktrue(t_2)$ rather than $g(t_2)$,
and, subtracting unity,
\begin{equation}
\E_\Pi\!\big[e^{-\Delta s}\big]-1
=\sum_{j}k_j(t_2)\,\frac{\ktrue_j(t_2)-g_j(t_2)}{g_j(t_2)}.
\label{eq:mismatch}
\end{equation}
Writing the right-hand side as an inner product of the vectors
$(\ktrue_j-g_j)/\sqrt{g_j}$ and $k_j(t_2)/\sqrt{g_j}$ and applying the
Cauchy--Schwarz inequality yields our third main result.

\begin{theorem}[Reference-mismatch bound]\label{res:chi2}
The $\chi^2$ divergence between the true and the assumed reference at $t_2$ is
lower-bounded by the observed violation of the fluctuation theorem,
\begin{equation}
\Dchi\!\big(\ktrue(t_2)\,\big\|\,g(t_2)\big)
\ \ge\
\frac{\big(\E_\Pi[e^{-\Delta s}]-1\big)^{2}}
{\displaystyle\sum_{j}k_j^{2}(t_2)/g_j(t_2)},
\label{eq:chi2}
\end{equation}
with $\Dchi(x\|y):=\sum_j (x_j-y_j)^2/y_j$, and requires no independent
knowledge of the true reference $\ktrue$.
\end{theorem}
This result runs opposite to the usual role of the $\chi^{2}$ divergence in
Monte Carlo and importance sampling, where a \emph{known} $\chi^{2}(x\|y)$
\emph{upper}-bounds the degradation of an estimator built from a mismatched
proposal~\cite{ChatterjeeDiaconis}. Equation~\eqref{eq:chi2} instead turns an
\emph{observed} violation of \eqref{eq:ift} into a certified \emph{lower} bound
on an otherwise inaccessible divergence, with no independent knowledge of
$\ktrue$.

The bound bypasses entirely the need to physically prepare the reference state,
an operational advantage on platforms such as driven mesoscopic devices, quantum processors, or
complex biophysical systems, where the target reference ensemble can be
inaccurate~\cite{Proctor2020}, costly to prepare~\cite{Motta2020}, or simply
inaccessible~\cite{Haah2016,Dudko2008}.
Instead, it enables the mismatch between the assumed reference and its true dynamical evolution to be quantified from measurements performed on arbitrary, experimentally accessible states $k(t)$.
The probe state is moreover a tunable
diagnostic. Its denominator obeys the identity
$\sum_j k_j^2(t_2)/g_j(t_2)=1+\Dchi\!\big(k(t_2)\|g(t_2)\big)$, so that
\eqref{eq:chi2} reads
$\Dchi(\ktrue\|g)\ge(\E_\Pi[e^{-\Delta s}]-1)^2/[1+\Dchi(k(t_2)\|g(t_2))]$. The
r.h.s. of \eqref{eq:mismatch} is a $k(t_2)$-weighted average of the local
reference mismatch $(\ktrue_j-g_j)/g_j$: spreading the probe uniformly dilutes a
localized discrepancy against microstates where the reference is accurate,
whereas concentrating $k(t_2)$, through the choice of the initial probe and the
assumed dynamics, on microstates suspected of control errors or unmodeled noise
sharpens the numerator.
Because concentrating the probe modifies the
denominator, the tightest bound optimizes the ratio rather than the numerator
alone---a genuine trade-off that a judiciously engineered probe resolves,
exposing localized discrepancies that a uniformly spread probe would average
away.

Finally, although $\Dchi$ is not a distance, it induces one in the
small-deviation limit. Writing $\ktrue-g=:\epsilon v$ with $\sum_j v_j=0$ and
$\epsilon$ small, the $\chi^{2}$ divergence reduces to
$\Dchi(\ktrue\|g)=\epsilon^{2}\DRF(v,v)$, with
$\DRF(v,v):=\sum_j v_j^2/g_j$ the Rao--Fisher metric evaluated at
$g$~\cite{Rao1945,amari2000methods}; for the $\chi^{2}$ divergence this quadratic form is exact
along the ray $g+\epsilon v$. In this limit \eqref{eq:chi2} certifies a lower
bound on the Rao--Fisher distance between the assumed and true reference
trajectories, making explicit a feature that runs through the whole
construction: the reference $g$ is an active geometric object, setting the
metric with which nearby departures from the assumed dynamics are measured~\cite{crooks2007measuring}.

\emph{Example: a driven two-level system.---}%
The simplest setting in which a genuinely time-dependent reference arises is a
classical two-level system---the dissipative counterpart of Boltzmann's
two-level unit---with energy gap $\epsilon(t) \equiv E_+(t) - E_-(t)$ swept in time and in contact
with a thermal bath. Its populations obey a Pauli master equation (a classical
Markov-jump, or Liouville, generator) whose transition rates satisfy detailed
balance at inverse temperature $\beta$, so that at each instant the unique
stationary state is the instantaneous Gibbs distribution
$\mathcal{G}(t) := \left(1-\mathcal{G}_+(t), \mathcal{G}_+(t)\right)$, with $\mathcal{G}_+(t) = \left[1+e^{\beta E_+(t)}\right]^{-1}$. In this model, the excited-state population of $k(t)$ relaxes according to
$\dot k_+(t)=-\Gamma\,[\,k_+(t)-\mathcal{G}_+(t)\,]$, and a large relaxation rate $\Gamma$
relative to the sweep speed makes the state track $\mathcal{G}(t)$ adiabatically.
Preparing the system at $t_1$ in the stationary state and reading it out at
$t_2$, we form the reference-relative entropy production and average it over
the physical joint law~\eqref{eq:Pi}.

\begin{figure}[t]
\centering
\includegraphics[width=\linewidth]{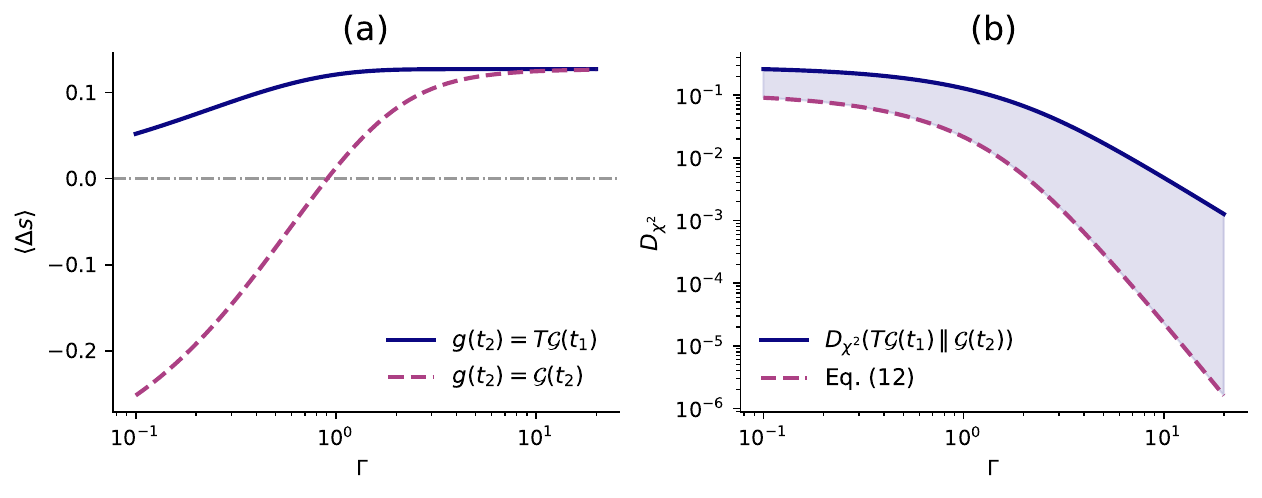}
\caption{Driven two-level system, with linearly decreasing energy gap $\epsilon(t) = \epsilon_0 + (\epsilon_1 - \epsilon_0) t$, $\epsilon_0 = 2, \,\epsilon_1 = 1/2$. Furthermore, we choose the temperature $\beta = 1$ and the initial distribution $k(t_1) = (1,0)$, i.e. the ground state; finally, without loss of generality (one could always rescale the rate $\Gamma$ to ensure this condition), we simply set the initial time $t_1 = 0$ and the final time $t_2=1$. a) Mean reference-relative entropy production
$\langle\Delta s\rangle$ versus relaxation rate $\Gamma$ (adiabaticity
increasing to the right), for the true propagated reference $g(t_2) = T \mathcal{G}(t_1)$
(solid) and for the instantaneous Gibbs reference $\mathcal{G}(t_2)$ (dashed). 
The
former obeys the second law exactly; the latter shows an apparent violation,
$\langle\Delta s\rangle<0$, that measures the non-adiabatic lag. (b) The true
reference mismatch $D_{\chi^2}(T \mathcal{G}(t_1)\|\mathcal{G}(t_2)$ (solid) and the lower bound
certified by the fluctuation-theorem violation through Theorem~\ref{res:chi2}
(dashed); the bound holds at every driving speed.}
\label{fig:qubit}
\end{figure}

Two choices of reference make the framework's content explicit
(Fig.~\ref{fig:qubit}). Taking $g(t_2)$ to be the \emph{true} reference $g_{\mathrm{true}}(t_2)$ obtained by
propagating $g(t_1) = \mathcal{G}(t_1)$ under the same channel obtained by solving the above master equation, Theorem~\ref{res:ift} holds exactly,
and with it the second law $\E_\Pi[\Delta s]\ge0$, at every driving speed
[Fig.~\ref{fig:qubit}(a), solid]. 
Taking instead the
\emph{instantaneous} Gibbs state $\mathcal{G}(t)$ to be reference, i.e. $g(t_2) = \mathcal{G}(t_2)$, the true propagated reference $g_\mathrm{true}(t)$ lags behind for any finite-speed driving, 
Theorem~\ref{res:ift} is violated, 
and $\E_\Pi[\Delta s]$ turns \emph{negative}
[Fig.~\ref{fig:qubit}(a), dashed], an apparent breakdown of the Second Law whose magnitude grows as the drive is made faster and vanishes in the
adiabatic limit. This is not a paradox but a diagnostic: by
Theorem~\ref{res:chi2}, the measured violation lower-bounds the
$\chi^2$ divergence between the true and the assumed reference
[Fig.~\ref{fig:qubit}(b)], certifying the nonadiabatic lag with no independent
characterization of the true dynamics. Dropping the external time dependence
returns $\mathcal{G}$ to the fixed Gibbs state and recovers the standard entropy
production and its ordinary fluctuation relation as the stationary limit of the
construction.

\emph{Conclusions.---}%
In this work we have shown that state convertibility relative to an arbitrary,
possibly time-dependent reference $g(t)$ is governed by $g(t)$-majorization,
which admits an exact dual description: each population vector maps to the
parent distribution of its relative populations, and convertibility reduces to
a one-dimensional convex-order condition, checkable through a martingale
coupling rather than through the construction and comparison of Lorenz curves
(Theorem~\ref{res:convert}). This unifies majorization and thermo-majorization
as the uniform- and Gibbs-reference instances of a single ordering and extends
both to references with no equilibrium or symmetry structure. Building on the
same construction, the reference-relative surprisal defines an entropy
production obeying an exact fluctuation theorem (Theorem~\ref{res:ift}), whose
average violation certifies, through a $\chi^{2}$-divergence bound, a lower
bound on the mismatch between an assumed and the true reference
(Theorem~\ref{res:chi2})---a diagnostic that needs no independent
characterization of the true dynamics and that we demonstrated on a driven
two-level system, where it converts an apparent violation of the second law
into a measure of nonadiabatic driving.

Two directions follow concretely. First, the fluctuation-theorem violation of
Theorem~\ref{res:chi2} is directly measurable: for a driven two-level system or
a Markov-jump network whose stationary state is only approximately known, a
single ensemble of forward trajectories fixes $\E_\Pi[e^{-\Delta s}]$ and hence
a certified lower bound on the reference error, a nonequilibrium alternative to
preparing and tomographing the stationary state. Second, because the
denominator of Theorem~\ref{res:chi2} equals $1+D_{\chi^2}(k(t_2)\|g(t_2))$, the
probe that maximizes the certified bound is itself determined by the reference,
defining an optimization---concentrating population where the reference is
least trusted---that is worth characterizing in full. 

\section*{Acknowledgments}
G.G. kindly acknowledges support from the Ministero dell’Università e della Ricerca (MUR) under the “Rita Levi-Montalcini” grant and from INFN.

\bibliography{bibliography}

\newpage
\onecolumngrid
\vspace{1em}
\begin{center}\rule{0.5\linewidth}{0.4pt}\end{center}
\vspace{0.5em}
\twocolumngrid

\section*{End Matter}
\emph{The Lorenz-curve criterion.---}%
In this End Matter we want to explicitly show how state convertibility admits an exact characterization directly in terms of the population vector,
in the classical language of majorization. Introduce the non-increasing step function
\begin{align}\label{eq:gammaEM}
    \Gamma_{\kappa}(\lambda;\,t) := \sum_{j = 1}^m \kappa_m^{\downarrow}(t) \,\mathds{1}_{\{\sum_{i=1}^{m-1}g_i^{\downarrow}(t) \leq \lambda < \sum_{i=1}^m g_i^{\downarrow}(t)\}}\,,
\end{align}
with $\lambda \in[0,1]$.
Here, the superscript $\downarrow$ indicates that the components are reordered according to the non-increasing ordering of the population vector $\kappa(t)$. More precisely, let $\pi_t$ be a permutation such that
\begin{align}
    \kappa_{\pi_t(1)}(t)
    \geq
    \kappa_{\pi_t(2)}(t)
    \geq \cdots \geq
    \kappa_{\pi_t(d)}(t).
\end{align}
Then, for any vector $v$, we define
\begin{align}
    v_m^\downarrow(t) := v_{\pi_t(m)}(t).
\end{align}
Thus, both $\kappa^\downarrow(t)$ and $g^\downarrow(t)$ are expressed in the ordering induced by $\kappa(t)$: the components of $\kappa(t)$ are arranged in non-increasing order, and the components of $g(t)$ are rearranged according to the same permutation.
 Its running integral
$L_\kappa(\Lambda;t):=\int_0^\Lambda \Gamma_\kappa(\lambda;t)\,d\lambda$ is the Lorenz
curve of $k$ taken relative to $g$: concave, non-decreasing, and pinned at
$L(0;t)=0$, $L(1;t)=1$. A transition matrix satisfying the admissibility
conditions~\eqref{eq:convertibility} exists if and only if~\cite{SM}
\begin{equation}
L_\kappa(\Lambda;t_1)
\ \ge\
L_\kappa(\Lambda;t_2)
\qquad\forall\,\Lambda\in[0,1],
\label{eq:lorenz}
\end{equation}
that is, the earlier reference Lorenz curve lies everywhere above the later one.
This criterion is the exact statement of \emph{majorization relative to the
reference $g$}---the classical notion of $d$-majorization or relative
majorization~\cite{Veinott1971,Ruch1980,Blackwell1953}, here made
time-dependent. It unifies and generalizes the two familiar cases in a single
inequality. When $g$ is uniform (the microcanonical case, microstates of equal
energy), the slabs have equal width and \eqref{eq:lorenz} reduces at the
breakpoints to the standard majorization condition
$\sum_{i\le\ell}\kappa^{\downarrow}_{i}(t_1)\ge\sum_{i\le\ell}\kappa^{\downarrow}_{i}(t_2)$.
When $g$ is a fixed Gibbs state (the canonical case), it becomes the
thermo-majorization criterion of the resource theory of
thermodynamics~\cite{Horodecki2013,brandao2013resource,gour2015resource,Renes2016,BuscemiGour2017}, in its exact,
non-approximated form. For a general and possibly time-dependent $g(t)$,
\eqref{eq:lorenz} extends both to references with no equilibrium or symmetry
structure.
The criterion is equivalent to Theorem~\ref{res:convert}: sorting-based Lorenz
dominance, $g(t)$-majorization, the convex order of parent distributions, and the
existence of a martingale coupling are one and the same relation. The
equivalence is proved in full detail in the Supplemental Material~\cite{SM}, where the
staircase construction of \eqref{eq:gammaEM} and the reference Lorenz curves are
also illustrated. Its practical drawback, and the reason it is not the criterion
we adopt in the main text, is that evaluating $\Gamma_\kappa$ requires sorting the relative populations at each time.

\widetext

\appendix

\begin{center}
\textbf{\large Supplementary Material}\\
\end{center}
\setcounter{equation}{0}
\setcounter{figure}{0}
\setcounter{table}{0}
\setcounter{page}{1}
\makeatletter
\renewcommand{\theequation}{S\arabic{equation}}
\renewcommand{\thefigure}{S\arabic{figure}}

In this Supplementary Material we provide all the detailed proofs of the necessity and sufficiency of the Lorenz Curve criterion, and of Theorems 1 and 3 of the main text.

\section{A --- Proof of the Necessity and Sufficiency of the Lorenz-Curve Criterion}

We hereby prove the following statement:
there exists a stochastic matrix $T$ with
\begin{align}\label{eq:T_conditions}
    Tg(t_1)=g(t_2)\wedge Tk(t_1)=k(t_2)
\end{align}
if and only if
\begin{align}\label{eq:lorenz_conditions}
    L_\kappa(\Lambda;t_1)\ \ge\ L_\kappa(\Lambda;t_2)\,\text{for all }\Lambda\in[0,1]\,,
\end{align}
starting from the following two lemmas.

\begin{lemma}
\label{lem:sorting-reduction}
Let $P_{t}\in\mathbb R^{m\times m}$ be the permutation matrix with $(P_tv)_j=v_{\pi_t(j)}$ for every vector $v$. If there exists a stochastic $\widetilde T$ with
\[
\widetilde T\, g^\downarrow(t_1)=g^\downarrow(t_2), \qquad \widetilde T\, k^\downarrow(t_1)=k^\downarrow(t_2),
\]
then $T:=P_{t_2}^{-1}\widetilde T P_{t_1}$ is stochastic and satisfies $Tg(t_1)=g(t_2)$, $Tk(t_1)=k(t_2)$.
\end{lemma}

\begin{proof}
Permutation matrices have exactly one $1$ per row and column, hence are stochastic, and $P_{t_2}^{-1}=P_{t_2}^{\mathsf T}$ is again a permutation matrix, hence stochastic. If $A,B$ are stochastic then so is $AB$: 
\[
\sum_i(AB)_{ij}
=\sum_i\sum_\ell A_{i\ell}B_{\ell j}
=\sum_\ell B_{\ell j}\bigg(\sum_iA_{i\ell}\bigg)
=\sum_\ell B_{\ell j}
=1.
\]
Hence $T=P_{t_2}^{-1}\widetilde TP_{t_1}$ is stochastic. Moreover $P_{t_1}g(t_1)=g^\downarrow(t_1)$ by definition of $P_{t_1}$, so
\[
Tg(t_1)
=P_{t_2}^{-1}\widetilde T P_{t_1}g(t_1)
=P_{t_2}^{-1}\widetilde T g^\downarrow(t_1)
=P_{t_2}^{-1}g^\downarrow(t_2)
=g(t_2),
\]
using $P_{t_2}^{-1}g^\downarrow(t_2)=g(t_2)$ (the inverse rearrangement). The identity $Tk(t_1)=k(t_2)$ follows identically.
\end{proof}

By Lemma~\ref{lem:sorting-reduction} it suffices to prove the theorem under the standing assumption that $g(t_1),k(t_1),g(t_2),k(t_2)$ are \emph{already} arranged in the order induced by their own $\kappa$, i.e.\ $\kappa_1(t)\ge\cdots\ge\kappa_m(t)$ for $t=t_1,t_2$; we adopt this assumption for the remainder of the proof and drop the $\downarrow$ superscripts.

\begin{lemma}
\label{lem:finite-checkpoints}
Condition~\eqref{eq:lorenz_conditions} holds for all $\Lambda\in[0,1]$ if and only if it holds at every point of the common refinement
\[
\left\{\sum_{j=1}^r g_j(t_1)\right\}_{r=0}^m
\cup
\left\{\sum_{j=1}^r g_j(t_2)\right\}_{r=0}^m.
\]
\end{lemma}

\begin{proof}
Necessity is trivial. For sufficiency: on each interval between two consecutive points of the combined breakpoint set, both $L_\kappa(\cdot\,;t_1)$ and $L_\kappa(\cdot\,;t_2)$ are linear (each is linear between its \emph{own} breakpoints, and the combined set refines both), so their difference is affine on that interval. An affine function that is $\ge 0$ at both endpoints of an interval is $\ge 0$ throughout the interval, since it attains its extrema at the endpoints. Hence non-negativity at the finitely many combined breakpoints implies non-negativity everywhere.
\end{proof}

We proceed in two steps: first, we prove that the existence of a stochastic matrix satisfying Eq.~\eqref{eq:T_conditions} implies the Lorenz-curve condition~\eqref{eq:lorenz_conditions}; subsequently, we prove the converse implication.

\subsubsection*{Step 1: Proving \eqref{eq:T_conditions} $\Rightarrow$ \eqref{eq:lorenz_conditions}}

\begin{lemma}[Variational form of $L_\kappa$]
\label{lem:variational}
For every $t\in\{t_1,t_2\}$ and $\Lambda\in[0,1]$,
\begin{align}
\label{eq:LP}
L_\kappa(\Lambda;t)
=
\max\Big\{
\sum_{j=1}^m\theta_j k_j(t)
:\,
\theta\in[0,1]^m,\ 
\sum_{j=1}^m\theta_j g_j(t)\le\Lambda
\Big\}.
\end{align}
\end{lemma}

\begin{proof}
This is the classical fractional knapsack problem, with item $j$ having weight $g_j(t)$ and value $k_j(t)=\kappa_j(t)g_j(t)$. Since the ratios of value to weight are ordered as
\[
\kappa_1(t)\ge\cdots\ge\kappa_m(t),
\]
the optimal strategy is to select the items in decreasing order of $\kappa_j(t)$, taking each item fully until the capacity $\Lambda$ is reached and, if necessary, taking a fractional amount of the final item. This is precisely the construction of $L_\kappa(\Lambda;t)$.
\end{proof}

\begin{proposition}
\label{prop:necessity}
If a stochastic matrix $T$ exists such that
\[
Tg(t_1)=g(t_2),
\qquad
Tk(t_1)=k(t_2),
\]
then \eqref{eq:lorenz_conditions} holds for all $\Lambda\in[0,1]$.
\end{proposition}

\begin{proof}
Fix $\Lambda\in[0,1]$, and let $\theta^*\in[0,1]^m$ be a maximizer in \eqref{eq:LP} for $L_\kappa(\Lambda;t_2)$, so that
\begin{align}
L_\kappa(\Lambda;t_2)
&=
\sum_{i=1}^m\theta_i^*k_i(t_2),
&
\sum_{i=1}^m\theta_i^*g_i(t_2)
&\le\Lambda.
\end{align}
Define
\begin{align}
\phi_j:=\sum_{i=1}^m\theta_i^*T_{ij},
\qquad j\in[m].
\end{align}
Since $T_{ij}\ge0$, $\sum_iT_{ij}=1$, and $\theta_i^*\in[0,1]$, each $\phi_j$ is a convex combination of the $\theta_i^*$ and therefore $\phi\in[0,1]^m$.

Using $Tk(t_1)=k(t_2)$, we obtain
\begin{align}
\sum_{i=1}^m\theta_i^*k_i(t_2)
&=
\sum_{i=1}^m\theta_i^*\sum_{j=1}^mT_{ij}k_j(t_1)
=
\sum_{j=1}^m\phi_jk_j(t_1).
\end{align}
Similarly, using $Tg(t_1)=g(t_2)$,
\begin{align}
\sum_{j=1}^m\phi_jg_j(t_1)
&=
\sum_{i=1}^m\theta_i^*\sum_{j=1}^mT_{ij}g_j(t_1)
=
\sum_{i=1}^m\theta_i^*g_i(t_2)
\le\Lambda.
\end{align}
Hence $\phi$ is feasible for the maximization problem \eqref{eq:LP} defining $L_\kappa(\Lambda;t_1)$. Therefore,
\begin{align}
L_\kappa(\Lambda;t_2)
=
\sum_{j=1}^m\phi_jk_j(t_1)
\le
L_\kappa(\Lambda;t_1).
\end{align}
Since $\Lambda$ was arbitrary, \eqref{eq:lorenz_conditions} follows.\\
\end{proof}
\vspace{0.5cm}

\subsubsection*{Step 2: Proving \eqref{eq:lorenz_conditions} $\Rightarrow$ \eqref{eq:T_conditions}}
\noindent
Assume \eqref{eq:lorenz_conditions}, let
\[
\{0=\lambda_0<\lambda_1<\cdots<\lambda_N=1\}
\]
be the union of the breakpoints
\[
\left\{G_j^\downarrow(t_1)\right\}_{j=0}^m
\cup
\left\{G_j^\downarrow(t_2)\right\}_{j=0}^m,
\qquad
G_j^\downarrow(t):=\sum_{i=1}^j g_i^\downarrow(t),
\]
and define
\[
\ell_r:=\lambda_r-\lambda_{r-1}>0,
\qquad r=1,\ldots,N.
\]
Thus $N\leq 2m-1$ and $\sum_{r=1}^N\ell_r=1$. For $\lambda\in(\lambda_{r-1},\lambda_r)$, define
\[
a_r:=\Gamma_\kappa(\lambda;t_1),
\qquad
b_r:=\Gamma_\kappa(\lambda;t_2).
\]
Since $\Gamma_\kappa(\cdot;t_1)$ and $\Gamma_\kappa(\cdot;t_2)$ are non-increasing, both $(a_r)_{r=1}^N$ and $(b_r)_{r=1}^N$ are non-increasing. Moreover, defining
\[
J_j:=\left\{r:
(\lambda_{r-1},\lambda_r)
\subset
[G_{j-1}^\downarrow(t_1),G_j^\downarrow(t_1))
\right\},
\]
and
\[
\mathcal I_j:=\left\{r:
(\lambda_{r-1},\lambda_r)
\subset
[G_{j-1}^\downarrow(t_2),G_j^\downarrow(t_2))
\right\},
\]
the families $\{J_j\}_{j=1}^m$ and $\{\mathcal I_j\}_{j=1}^m$ partition $[N]$, and
\begin{align}
\sum_{r\in J_j}\ell_r&=g_j(t_1),
&
\sum_{r\in\mathcal I_j}\ell_r&=g_j(t_2),
\\
a_r&=\kappa_j(t_1)\quad(r\in J_j),
&
b_r&=\kappa_j(t_2)\quad(r\in\mathcal I_j).
\end{align}
By Lemma~\ref{lem:finite-checkpoints}, condition \eqref{eq:lorenz_conditions} is equivalent to
\begin{align}
\label{eq:atomic-majorization}
\sum_{r=1}^q\ell_r b_r
\leq
\sum_{r=1}^q\ell_r a_r,
\qquad q=1,\ldots,N,
\end{align}
with equality for $q=N$.
We now use the following weighted version of the Hardy--Littlewood--Pólya theorem.

\begin{lemma}[Weighted Hardy--Littlewood--Pólya theorem]
\label{lem:weighted-HLP}
Let $\ell_r>0$ with $\sum_r\ell_r=1$, and let $a,b\in\mathbb R^N$ be non-increasing. Then there exists a non-negative matrix $D\in\mathbb R^{N\times N}$ satisfying
\begin{align}
\sum_{r=1}^N D_{rs}=1,
\qquad
D\ell=\ell,
\qquad
D(\ell\odot a)=\ell\odot b,
\end{align}
if and only if
\begin{align}
\sum_{r=1}^q\ell_rb_r
\leq
\sum_{r=1}^q\ell_ra_r,
\qquad q=1,\ldots,N,
\end{align}
with equality for $q=N$.
\end{lemma}
The lemma is the standard weighted form of the Hardy--Littlewood--Pólya characterization of majorization: the inequalities above are precisely the condition that the weighted vector $b$ is majorized by the weighted vector $a$. Applying Lemma~\ref{lem:weighted-HLP} to the sequences constructed above, we obtain a non-negative matrix $D$ such that
\begin{align}
\sum_{r=1}^N D_{rs}=1,
\qquad
D\ell=\ell,
\qquad
D(\ell\odot a)=\ell\odot b.
\end{align}
Define
\begin{align}
\label{eq:folding}
T_{ij}:=
\frac{1}{g_j(t_1)}
\sum_{r\in J_j}
\sum_{s\in\mathcal I_i}
\ell_rD_{sr},
\qquad i,j\in[m].
\end{align}
Such a matrix satisfies the following proposition.
\begin{proposition}
\label{prop:refolding}
The matrix $T$ defined by \eqref{eq:folding} is stochastic and satisfies
\[
Tg(t_1)=g(t_2),
\qquad
Tk(t_1)=k(t_2).
\]
\end{proposition}

\begin{proof}
Non-negativity is immediate. Moreover,
\begin{align}
\sum_{i=1}^mT_{ij}
&=
\frac{1}{g_j(t_1)}
\sum_{r\in J_j}\ell_r
\sum_{i=1}^m\sum_{s\in\mathcal I_i}D_{sr}
\nonumber\\
&=
\frac{1}{g_j(t_1)}
\sum_{r\in J_j}\ell_r
\sum_{s=1}^ND_{sr}
=
\frac{1}{g_j(t_1)}
\sum_{r\in J_j}\ell_r
=1,
\end{align}
where we used the column-stochasticity of $D$ and
$\sum_{r\in J_j}\ell_r=g_j(t_1)$. Hence $T$ is stochastic.
For the $g$-transport property,
\begin{align}
\sum_{j=1}^mT_{ij}g_j(t_1)
&=
\sum_{j=1}^m
\sum_{r\in J_j}
\sum_{s\in\mathcal I_i}
\ell_rD_{sr}
\nonumber\\
&=
\sum_{s\in\mathcal I_i}
\sum_{r=1}^N D_{sr}\ell_r
=
\sum_{s\in\mathcal I_i}(D\ell)_s
=
\sum_{s\in\mathcal I_i}\ell_s
=
g_i(t_2).
\end{align}
Thus $Tg(t_1)=g(t_2)$.
For the $k$-transport property, observe that for $r\in J_j$,
\[
a_r=\kappa_j(t_1),
\]
and hence
\[
\ell_ra_r
=
\ell_r\kappa_j(t_1).
\]
Using $k_j(t_1)=\kappa_j(t_1)g_j(t_1)$ and \eqref{eq:folding},
\begin{align}
\sum_{j=1}^mT_{ij}k_j(t_1)
&=
\sum_{j=1}^m
\sum_{r\in J_j}
\sum_{s\in\mathcal I_i}
\ell_ra_rD_{sr}
\nonumber\\
&=
\sum_{s\in\mathcal I_i}
\sum_{r=1}^ND_{sr}\ell_ra_r
=
\sum_{s\in\mathcal I_i}
\big[D(\ell\odot a)\big]_s
\nonumber\\
&=
\sum_{s\in\mathcal I_i}\ell_sb_s.
\end{align}
Since $b_s=\kappa_i(t_2)$ for $s\in\mathcal I_i$,
\begin{align}
\sum_{s\in\mathcal I_i}\ell_sb_s
&=
\kappa_i(t_2)\sum_{s\in\mathcal I_i}\ell_s
=
\kappa_i(t_2)g_i(t_2)
=
k_i(t_2).
\end{align}
Therefore $Tk(t_1)=k(t_2)$.
\end{proof}

This proves \eqref{eq:T_conditions}, completing the proof of the sufficiency direction.
\vspace{0.5 cm}

\subsubsection*{Conclusion of the proof}
By Proposition~\ref{prop:necessity}, the existence of a stochastic matrix $T$ satisfying
\[
Tg(t_1)=g(t_2),
\qquad
Tk(t_1)=k(t_2)
\]
implies
\[
L_\kappa(\Lambda;t_1)\geq L_\kappa(\Lambda;t_2)
\qquad\text{for all }\Lambda\in[0,1].
\]
Conversely, assuming \eqref{eq:lorenz_conditions}, Lemma~\ref{lem:finite-checkpoints} reduces the Lorenz-curve condition to the finite set of atomic majorization inequalities. Lemma~\ref{lem:weighted-HLP} then yields an $\ell$-doubly-stochastic matrix $D$ satisfying
\[
D(\ell\odot a)=\ell\odot b,
\]
and Proposition~\ref{prop:refolding} folds this atomic transformation into a stochastic matrix $T$ satisfying
\[
Tg(t_1)=g(t_2),
\qquad
Tk(t_1)=k(t_2)
\]
in the sorted coordinates. Finally, Lemma~\ref{lem:sorting-reduction} transports $T$ back to the original coordinates while preserving both identities. Hence \eqref{eq:T_conditions} and \eqref{eq:lorenz_conditions} are equivalent, proving Theorem~3.

\clearpage
\section{B --- Proof of Theorem 1}
\noindent
Let $P(j_1,j_2)$ denote a joint distribution of the pair $(j_1,j_2)$, with marginals
\begin{align}
    P(j_1) &= \sum_n P(j_1,n),
    &
    P(j_2) &= \sum_n P(n,j_2),
\end{align}
and define
\begin{equation}
    T_{j_2j_1}
    := \frac{P(j_1,j_2)}{g_{j_1}(t_1)}.
\end{equation}
We establish the three properties required for $T$ to be a stochastic transformation satisfying
$Tg(t_1)=g(t_2)$ and $Tk(t_1)=k(t_2)$.
First, summing over the first index gives
\begin{align}
    \sum_n T_{nj_1}
    = \sum_n \frac{P(j_1,n)}{g_{j_1}(t_1)}
    = \frac{P(j_1)}{g_{j_1}(t_1)}.
\end{align}
Therefore,
\begin{align}
    T \text{ is stochastic}
    \iff
    P(j_1)=g_{j_1}(t_1).
\end{align}
Second, we have
\begin{align}
    \sum_n T_{j_2n}g_n(t_1)
    &= \sum_n P(n,j_2)
    = P(j_2).
\end{align}
Hence,
\begin{align}
    Tg(t_1)=g(t_2)
    \iff
    P(j_2)=g_{j_2}(t_2).
\end{align}
Finally, using $k_n(t_1)=g_n(t_1)\kappa_n(t_1)$, we obtain
\begin{align}
    \sum_n T_{j_2n}k_n(t_1)
    &= \sum_n P(n,j_2)\kappa_n(t_1)
    \\
    &= P(j_2)\,
    \mathbb{E}\!\left[
        x_1 \,\middle|\, x_2=\kappa_{j_2}(t_2)
    \right].
\end{align}
Consequently,
\begin{align}
    Tk(t_1)=k(t_2)
    \iff
    P(j_2)\,
    \mathbb{E}\!\left[
        x_1 \,\middle|\, x_2=\kappa_{j_2}(t_2)
    \right]
    = k_{j_2}(t_2).
\end{align}
Combining these three results, a stochastic transformation $T$ satisfying
\begin{align}
    Tg(t_1)=g(t_2),
    \qquad
    Tk(t_1)=k(t_2)
\end{align}
exists if and only if there exists a joint distribution $P(j_1,j_2)$ with marginals
\begin{align}
    P(j_1)=g_{j_1}(t_1),
    \qquad
    P(j_2)=g_{j_2}(t_2),
\end{align}
such that
\begin{align}
    \mathbb{E}\!\left[
        x_1 \,\middle|\, x_2=\kappa_{j_2}(t_2)
    \right]
    =\kappa_{j_2}(t_2).
\end{align}
In other words, $P(j_1,j_2)$ defines a martingale coupling between the distributions $p_\kappa(x_1,t_1)$ and $p_\kappa(x_2,t_2)$.

\clearpage
\section{C --- Proof of Theorem 3}
\noindent
We begin by evaluating the exponential average of the entropy difference:
\begin{align}
    \mathbb{E}_\Pi\!\left[e^{-\Delta s}\right]
    &= \sum_{j_1,j_2}
    k_{j_1}(t_1) T_{j_2 j_1}
    \frac{\kappa_{j_2}(t_2)}{\kappa_{j_1}(t_1)}
    \\
    &= \sum_{j_1,j_2}
    T_{j_2 j_1} g_{j_1}(t_1)\kappa_{j_2}(t_2)
    \\
    &= \sum_{j_2}
    g^\mathrm{true}_{j_2}(t_2)\kappa_{j_2}(t_2)
    \\
    &= \sum_{j_2}
    k_{j_2}(t_2)
    \frac{g^\mathrm{true}_{j_2}(t_2)}
    {g_{j_2}(t_2)}.
\end{align}
Consequently,
\begin{align}
    \mathbb{E}_\Pi\!\left[e^{-\Delta s}\right]-1
    =
    \sum_{j_2}
    \frac{
        k_{j_2}(t_2)
        \left[
            g^\mathrm{true}_{j_2}(t_2)-g_{j_2}(t_2)
        \right]
    }{
        g_{j_2}(t_2)
    }.
\end{align}
Applying the Cauchy--Schwarz inequality then yields
\begin{align}
    \left(
        \mathbb{E}_\Pi\!\left[e^{-\Delta s}\right]-1
    \right)^2
    \leq
    \left[
        \sum_{j_2}
        \frac{k_{j_2}^2(t_2)}
        {g_{j_2}(t_2)}
    \right]
    \left[
        \sum_{j_2}
        \frac{
            \left(
                g^\mathrm{true}_{j_2}(t_2)
                -g_{j_2}(t_2)
            \right)^2
        }{
            g_{j_2}(t_2)
        }
    \right].
\end{align}
Recognizing the second factor as the $\chi^2$-divergence,
\begin{align}
    D_{\chi^2}\!\left(
        g^\mathrm{true}(t_2)
        \,\middle\|\,
        g(t_2)
    \right)
    \geq
    \frac{
        \left(
            \mathbb{E}_\Pi\!\left[e^{-\Delta s}\right]-1
        \right)^2
    }{
        \displaystyle
        \sum_j
        k_j^2(t_2)/g_j(t_2)
    },
\end{align}
which proves the stated result.


\end{document}